\documentclass[a4paper,11pt]{article}
\usepackage{amssymb,amsmath,amsthm}
\usepackage{marvosym}
\usepackage{graphicx}

\usepackage{amsfonts}
\usepackage{latexsym}
\usepackage{color}
\usepackage[english]{babel}
\usepackage{hyperref}
\usepackage{float}

\numberwithin{equation}{section}

\newtheorem{lemma}{Lemma}[section]
\newtheorem{theorem}[lemma]{Theorem}

\newtheorem{corollary}[lemma]{Corollary}
\newtheorem{remark}[lemma]{Remark}

\newtheorem{definition}[lemma]{Definition}
\def\sleq{\lesssim}

\newcommand{\tb}{{\tt b}}

\def\as{p}
\def\fhi{x}
\def\exte{[\Pi^{(s)}_M(\as)]^{ext}}
\def\exten#1{[\Pi^{(s)}_M(#1)]^{ext}}

\newcommand{\R}{\mathbb R}

\newcommand{\Z}{\mathbb Z}

\newcommand{\N}{\mathbb N}
\newcommand{\T}{\mathbb T}

\def\norm#1{\|#1  \|}

\def\im{{\rm i}}

\def\timereg{C^{\infty}_b}

\def\td{2}
\def\sy#1{S^{#1}_\delta}
\def\poi#1#2{\left\{#1;#2\right\}}

\def\cZ{{\mathcal{Z}}}

\definecolor{awesome}{rgb}{1.0, 0.13, 0.32}
\definecolor{darkgr}{rgb}{0.0, 0.62, 0.42}
\definecolor{cyan}{rgb}{0.0, 0.72, 0.92}

\def\cF{{\mathcal{F}}}

\def\cT{{\mathcal{T}}}

\def\QS{{\mathcal{Q}\kern-0.3pt\mathcal{S}}}

\def\tC{{\tt C}}

\def\tD{{\tt D}}
\def\tR{{\tt R}}

\def\cO{{\mathcal{O}}}

\def\cR{{\mathcal{R}}}
\def\cB{{\mathcal{B}}}
\def\tb{{\tt b}}

\newcommand{\Span}{\mathrm{span}}
\newcommand{\scala}[2]{{#1} \cdot {#2} }

\def\mom{1}

\def\da{{\frak{a}}}

\begin{document}

\title{A global Nekhoroshev theorem for particles on the torus with
  time dependent Hamiltonian. }


\date{}


\author{ Dario Bambusi\footnote{Dipartimento di Matematica, Universit\`a degli Studi di Milano, Via Saldini 50, I-20133
Milano. 
 \textit{Email: } \texttt{dario.bambusi@unimi.it}}}

\maketitle

\begin{abstract}
We prove a $C^\infty$ global version of Nekhoroshev theorem for time
dependent Hamiltonians in $\R^d\times\T^d$. Precisely, we prove a
result showing that for all times the actions of the unperturbed
systems are bounded by a constant times $ \langle
t\rangle^{\epsilon}$.  We apply the result to the dynamics of a
charged particle in $\T^d$ subject to a time dependent electromagnetic
field.
\end{abstract}
\noindent


\medskip


 
\section{Introduction}\label{intro}

In this paper we study the dynamics of a Hamiltonian system of the form
\begin{equation}
  \label{hami}
H=h_0(\as )+P(\as ,\fhi,t)\ ,\quad
h_0(\as):=\sum_{j=1}^d\frac{p_j^2}{2}\ ,\quad (\as,\fhi)\in\R^d\times
\T^d\ ,
\end{equation}
where 
$P(\as ,\fhi,t)$ is a $C^{\infty}$ function bounded by
$\langle\as\rangle^\tb$ with some $\tb<\td$ (as usual $\langle\as\rangle:=\sqrt{1+\left\|\as\right\|^2}$).  The main example we have in mind is that
of a particle subject to a time dependent electromagnetic field,
namely
\begin{equation}
  \label{ele}
H=\sum_{j=1}^d\frac{\left(p_j-A_j(x,t)\right)^2}{2}+\varphi(x,t)\ ,
\end{equation}
with $A(x,t)$ and $\varphi(x,t)$ functions of class $C^\infty(\T^d\times\R;\R)$.

We are going to prove that $\forall \epsilon>0$ $\exists C_\epsilon$,
s.t. the
solution of the system fulfils forever an estimate of the form
\begin{align}
   \label{conclusion}
\left\|\as (t)\right\|&\leq C_\epsilon\langle
t\rangle^{\epsilon}\left(1+\left\|\as _0\right\|\right)\ ,\quad\forall t\in\R
\\
 \langle
t\rangle&:=\sqrt{1+|t|^2}\ .
\end{align}
This is essentially a Nekhoroshev type theorem \cite{Nek1,Nek2} (see
also \cite{BGG85,BG86,poschel,GioPisa,chierchia,kunze} and literature therein) and we
give here a complete proof of it. Such a proof is a transposition to
the classical context of the proof of a quantum Nekhoroshev theorem
that has been given in \cite{BLM22a,BLM22b,QN}, and turns out to be
particularly simple in the classical case.

We remark that results controlling the growth of the Sobolev norms of
the wave function by $\langle t\rangle^\epsilon$ are by now quite
standard in a quantum context
\cite{Bou2,Bo97a,wang08,delort,BDGM2,BLM22a,BM19}, and the present
work arises from the curiosity of understanding is something similar
is true in classical mechanics. Actually we recall that standard
Nekhoroshev results only allow to get informations over finite amounts
of time, while we are here interested in an infinite time scale.

\vskip5pt

We recall that original Nekhoroshev's theorem deals with perturbations
of integrable systems and gives upper bounds on the drift of the
actions for very long times. Original Nekhoroshev's theorem deals with
an analytic context and the result it gives is valid over times which
are exponentially long with the inverse of the size $\epsilon$ of the
perturbation. Corresponding $C^{\infty}$ versions have been obtained
\cite{MS02,MS04,Bou10,BL21,ultra}: in these cases one gets bounds valid over a
time scale of order $\epsilon^{-N}$ with an arbitrary $N$. Our result
is in this line.

As a variant with respect to the classical framework, in this paper there
is no $\epsilon$, the reason is that we are interested in large
solutions, so the smallness is related to the decay properties of
Hamiltonians at infinity. To completely exploit this idea, we modify one
of the pillars on which Nekhoroshev theorem is build, namely we change the
definition of resonant region: the idea is that  a point $\as$ is resonant,
not if $\left|\omega(\as)\cdot k\right|<\alpha$ with some small
$\alpha$, but if $\left|\omega(\as)\cdot k\right|<\left\|k\right\|
\left\|\as\right\|^\delta$ with some positive $\delta<1$. This has the
advantage that small denominators are actually large in the
nonresonant regions. Of course, in order to exploit this idea, one has
to considerably change the classical construction of the geometric
part of the proof of Nekhoroshev's theorem, but fortunately this turns
out to be particularly simple. In this framework also the so called
analytic part of the proof is particularly simple and the tools of
symbolic calculus developed in the framework of pseudodifferential
calculus turn out to be very efficient in this context. 

A further difference with respect to the classical proof is that,
following \cite{BL21}, we perform here a normal form procedure
globally defined on the phase space and this simplify considerably the
geometric part of the proof. 

The final step leading to the control of the solution over an infinite
time scale is obtained by iterating the application of Nekhoroshev's
theorem.  
\vskip 5pt

We recall that a result very closely related to the present one is
that obtained by Giorgilli and Zehnder in \cite{GZ92}. In that paper the
authors considered a system of the form \eqref{hami} with an
analytical and globally bounded perturbation $P$. They proved that if
the initial datum is large enough, then it takes a time exponentially
long with its size to possibly double the size of the actions. This is
the analogue of our Theorem \ref{nekho}. The extension we get is that
we allow the perturbation to be unbounded and $C^\infty$, but we think
that the main point is that the present proof is much simpler than
that in \cite{GZ92}. We also recall the result by Bounemoura
\cite{BenTime}, which however deals with the more complicated situation
where the small denominators involve also the time dependence. 

Finally we recall that there exists a different proof, by Lochak
\cite{Loc,Loc2}, of Nekhoroshev's Theorem, a proof which is much
simpler than the original one, and which has also proved to be
suitable for the extension to some infinite dimensional systems
\cite{BG93} and in particular to PDEs \cite{Bam99,BG23}. Unfortunately
we have not been unable to adapt such a proof to the present time
dependent case.

{\it Acknowledgements} The present research was founded by the PRIN project
  2020XB3EFL Hamiltonian and dispersive PDEs. It was also supported by
GNFM.

\section{ Main Result}\label{state}

\begin{definition}
	\label{psM}
	A function $f\in C^{\infty}(\T^d\times \R^d\null)$ is said to
	be a symbol of class $S^m_{\delta}$, if it fulfils
	\begin{equation}
	\label{psM.1}
	\left|\partial^\alpha_\fhi\partial^\beta_\as  f(\fhi,\as )\right|\leq
	C_{\alpha,\beta}\langle \as \rangle^{m-\delta|\beta|}\ ,\quad \forall
	\alpha,\beta\in\N^n\,, \quad \forall (\fhi,\as ) \in \T^d\times\R^d\null\ .
	\end{equation}
\end{definition}

The best constants s.t. \eqref{psM.1} hold are a family of seminorms
for the space of symbols.  In this way the
space of symbols becomes a Fr\'echet space.

In order to deal with time dependent perturbations, we have to
consider also function taking value in the spaces of symbols. 

\begin{definition}\label{def.smooth}
	If $\cF$ is a Fr\'echet space, we denote by
        $C^k_b\left(\R; \cF\right)$ the space of the functions $f\in
        C^{k}\left(\R; \cF\right)$, such that all the seminorms of
        $\partial_t^jf$ are bounded uniformly over $\R$ for all $j\leq
        k$. If this is true for all $k$ we write $f\in
        C^\infty_b\left(\R; \cF\right)$ .
\end{definition}

\begin{theorem}
  \label{main}
Assume that $P\in C^{\infty}_b(\R;S^\tb_1)$
with $\tb<\td$, then $\forall \epsilon>0$ 
$\exists R_\epsilon$, s.t., if the initial datum
fulfills $\left\|\as_0\right\|\geq R_\epsilon$ then along the solutions of the Cauchy problem for the Hamilton
equation of \eqref{hami} one has
\begin{equation}
  \label{main.res}
\left\| \as (t)\right\|\leq 16\left\|\as_0\right\|
\left\langle\frac{t}{\left\|\as_0\right\|}\right\rangle^\epsilon
 \ ,\quad \forall t\in\R\ .
  \end{equation}
\end{theorem}

Then, by compactness of the ball of radius $R_\epsilon$, it is easy to
obtain the following corollary
\begin{corollary}
  \label{main.co}
Assume that $P\in C^{\infty}_b(\R;S^\tb_1)$
with $\tb<\td$, then $\forall \epsilon>0$ 
$\exists C_\epsilon$, s.t. along the solutions of the Cauchy problem for the Hamilton
equation of \eqref{hami} with initial datum $\as_0$ one has
\begin{equation}
  \label{main.res}
\left\| \as (t)\right\|\leq C_\epsilon \left(1+\left\|\as_0\right\|\right)
\left\langle {t}\right\rangle^\epsilon
 \ ,\quad \forall t\in\R\ .
  \end{equation}
\end{corollary}

Such  results have to be confronted with the corresponding quantum
result, in which one considers the Schr\"odinger equation with
Hamiltonian given by the quantisation of $H$. In this case it was
proved in \cite{BLM22b,QN} that the wave function $\psi(t)$ fulfils
$$
\left\|\psi(t)\right\|_{H^s}\leq
C_{\epsilon,s}\left\|\psi_0\right\|_{H^s}\langle
t\rangle^{\epsilon}\ ,\quad \forall t\in\R\ .
$$
Actually the question of the validity of a similar estimate in the
classical case was the main motivation for the present work.

As anticipated above, the main example we have in mind is that of a
particle in an electromagnetic field with Hamiltonian \eqref{ele}, in
which, in particular one has
$$
P(\as,\fhi,t)=-\as\cdot
A(\fhi,t)+\frac{\left\|A(\fhi,t)\right\|^2}{2}+\varphi(\fhi,t)\ .
$$
We thus get that maybe the energy of the particle grows to infinity,
but the average power at which energy is transferred to the particle
decrees faster than any power of time.

\section{Analytic Part}\label{NF.1}

We start by fixing some notations and definitions that will be used in
the rest of the paper. Given two real valued functions $f$ and $g$,
sometimes we will use the notation $f \lesssim g$ to mean that there
exists a constant $C>0$, independent of all the relevant quantities,
such that $f \leq C g$. If $f \lesssim g$ and $g \lesssim f$, we will
write $f \simeq g$.

 We will denote by $B_R(\as)$ the open
ball of radius $R$ centered $\as$.

Given a function $g$, we will denote by $X_g$ the corresponding
Hamiltonian vector field and by
\begin{equation}
  \label{poisson}
\left\{f;g\right\}:=dfX_g\equiv \sum_{j=1}^d\frac{\partial
  f}{\partial\fhi_j}\frac{\partial g}{\partial \as _j}-\frac{\partial
  g}{\partial\fhi_j}\frac{\partial f}{\partial \as _j} 
\end{equation}
the Poisson bracket of two functions. Remark that if $f\in\sy{m_1}$
and $g\in\sy{m_2}$, then
$\left\{f;g\right\}\in\sy{m_1+m_2-\delta}$. The gain of $\delta$ is
fundamental for the construction of the regularising transformation.

Given $\mu>0$ and $0<\delta<1$ (typically $\mu\ll 1$ and $\delta \simeq 1$), we give the following definitions:
\begin{definition}
	\label{res}
	We say that a point $\as \in \R^d\null$ is resonant with
	$k\in\Z^d\setminus\{0\} $ if
	\begin{equation}
	\label{reso.1}
	|\as\cdot k| \leq \|\as \|^\delta \|k\| \quad \text{and}\quad \|k\| \leq \|\as \|^\mu\,.
	\end{equation}
\end{definition}

\begin{definition}\label{def.nf}[Normal form]
	We say that a function $Z(\as ,\fhi)=\sum_{k\in\Z^d}\hat Z(\as )e^{ik\cdot\fhi}$ is \emph{in normal
		form} if all the points in supp($\hat Z_k(.))$ are
        resonant with $k$.
        \\
        \noindent We say that a function $Z(\as ,\fhi,t)$ is in resonant
        normal form if this is true in the above sense, for any fixed time $t$.
\end{definition}

In order to characterise the properties of the Lie transform we need
to introduce also the following class of functions
\begin{definition}
  \label{resti}
A  function $R$ will be said to be a remainder of order $N$ and we
will write $R\in\cR^N$ if  
\begin{align}
	\label{resti.1}
	\left|\partial^\alpha_\fhi\partial^\beta_\as  R(\fhi,\as )\right|\leq
	C_{\alpha,\beta}\langle \as \rangle^{-N}\ ,\quad\\ \forall
	\alpha,\beta\in\N^n\, ,\ |\alpha|+|\beta|\leq 2\ , 
         \forall (\fhi,\as ) \in \T^d\times\R^d\null\ .
\end{align}
\end{definition}

In the following we will use time dependent
transformations $\Phi(\as ,\fhi,t)$ with the property that for any fixed $t$ they are
canonical. In this case it is easy to see that the change of
coordinates $(\as,\fhi)=\Phi(\as',\fhi',t)$ transforms the equations
of motions of a Hamiltonian $H$ to the equations of motion of a new
Hamiltonian $H'$. In this case we will say that $\Phi$ conjugates $H$
and $H'$.

We are going to prove the following normal form theorem

\begin{theorem}\label{norm.form}[Normal Form]
	Let $H$ be given by \eqref{hami}, with $P\in
        \timereg\left(\R;S^{\tb}_1\right)$, $\tb<\td$. Fix $N\gg1$,
        then there exists
        $0<\delta_*<\mom$, $\mu_*>0$ such that, if
        $\delta_*<\delta<\mom$, $0<\mu<\mu_*$, then
there exists a time dependent canonical
        transformation $\cT$
 which conjugates $H$ to 
	\begin{equation}\label{forma.normale}
	H^{(N)}:=	h_0+Z_N(t) + R^{(N)}(t)\,,
	\end{equation}
	with
 $Z_N \in \timereg\left(\R; S^\tb_\delta\right)$ in normal form, while
$R^{(N)} \in \timereg\left(\R;\cR^{
	  N}\right)$ is a remainder of order $N$.
        
 Denoting $(\as,\fhi)=\cT(\as',\fhi',t)$, one has
                  $\left\|\as-\as'\right\|\leq C\left\|
 \as\right\|^{\tb-\delta}$.
\end{theorem}
The rest of this section is devoted to the proof of this theorem.

The idea is to perform a sequence of canonical transformations
conjugating the original system to a normal form plus a remainder whose
growth at infinity decreases at each step. Eventually it becomes a
function arbitrarily decreasing in the action space.

\subsection{Time dependent Lie transform}

The canonical transformations will be constructed as Lie transforms
generated by time dependent symbols. So, first we recall the main
definitions and properties.

Consider a 
family of time dependent Hamiltonians $g(\as ,\fhi,t)$, but think of $t$
as an external parameter. Denote by
$\Phi_g^\tau(\as ,\fhi, t)$ the time $\tau$ flow it generates, namely the
solution of
\begin{equation}
  \label{chitau}
\frac{d\fhi}{d\tau}=\frac{\partial g}{\partial \as }(\as ,\fhi,t)\ , \quad
\frac{d \as }{d\tau}=-\frac{\partial g}{\partial \fhi}(\as ,\fhi,t) \ .
\end{equation}
In the case $g\in\sy m$ with $m\leq 1$ the flow of \eqref{chitau} is
globally defined. This is the only situation we will encounter.

\begin{definition}
  \label{lie}
The time dependent coordinate transformation
\begin{equation}
  \label{timedep}
(\as ,\fhi)=\Phi_g(\as ',\fhi',t)\equiv
  \Phi^1_g(\as ',\fhi',t):=\left.\Phi_g^\tau(\as ',\fhi',t)\right|_{\tau=1}\
\end{equation}
is called the time dependent Lie transform generated by $g$.
\end{definition}

\begin{remark}
  \label{how}
The time dependent Lie transform $\Phi_g$ conjugates a Hamiltonian $f$
to the Hamilton
\begin{align}
  \label{timetras}
  f'(\as ',\fhi',t) :=f(\Phi_g(\as ',\fhi',t))-\Psi_g(\as ',\fhi',t)\ ,
\\
\label{psi}
\Psi_g(\as ',\fhi',t):=\int_0^1  \frac{\partial
  g}{\partial t}\left(\Phi_g^\tau(\as ',\fhi',t)\right) \, d\tau \ .
\end{align}
This can be easily seen by working in the extended phase space in
which time is added as a new variable.
\end{remark}

Given a function $f\in\sy {m}$, we study $f\circ\Phi_g$.  We start
by the time independent case

\begin{lemma}
  \label{resto.Lie}
  Let $g\in\sy{\eta}$ be a time independent function, and
  $f\in\sy{m}$, with $\eta<\delta$, then, for any
  positive $N$, one has
  \begin{equation}
    \label{lie.tot}
f\circ\Phi^1_g=\sum_{l=0}^N\frac{f_{l}}{l!}+R_{m-(N+1)(\delta-\eta)}\ ,
  \end{equation}
  with $f_l\in \sy {m-l(\delta-\eta)}$, precisely given by
\begin{equation}
  \label{seqlie}
f_{0}:=f\ ,\quad f_{l}:=\poi{f_{l-1}}{g}\equiv
\left.\frac{d^l}{dt^l}\right|_{t=0} f\circ\Phi^t_g\ , \quad l\geq1\ ,
\end{equation}
and $\cR_{m-(N+1)(\delta-\eta)}$ a remainder of order ${(N+1)(\delta-\eta)}-m$.
\\
Furthermore, if one denotes $(\as,\fhi)=\Phi_g(\as',\fhi',t)$, one has
\begin{equation}
  \label{defo}
\left\|\as-\as'\right\|\leq C\left\|\as\right\|^{\eta}\ .
  \end{equation}
\end{lemma}
\proof Just use  the formula for the remainder of the Taylor series (in
time) which gives
$$
f\circ\Phi^1_g=\sum_{l=0}^N\frac{f_{l}}{l!}+ \frac{1}{N!}\int_0^1
(1+s)^Nf_{N+1}\circ\Phi^s_gds\ ,
$$
from which the thesis immediately follows. 
\qed

In particular
we have the following corollary which covers the case of the time
dependent Lie transform and which is proved by simply remarking that
$\Psi_g\in C^\infty(\R;\sy\eta)$ up to a remainder of arbitrary order.

\begin{corollary}
  \label{time dep}
  Let $g\in C^\infty_b(\R;\sy{\eta})$ with $\eta<\delta$;  denote by
  $\Phi_g$ the time dependent Lie transform generated by $g$. Let 
  $f\in\sy{m}$, then, for any $N$,  one has
  \begin{equation}
    \label{lie.tot1}
    f\circ\Phi_g=f+\left\{f;g\right\}+C^{\infty}(\R;
    \sy{m-2(\delta-\eta)})+C^\infty(\R;\sy\eta)+C^\infty(\R;\cR^N)\ . 
  \end{equation}
 By  $C^\infty(\R;\sy\eta)$ in the above formula, we mean a function
 belonging to such a space and similarly for the other terms. 
\end{corollary}

From now on, time will only play the role of a parameter, so we will
omit to write explicitly this variable and omit to specify the dependence on
it, which will always be of class $C^\infty_b$.

We are now ready for the construction of the normal form
transformation. Before starting we change the
family of seminorms that we will use for symbols. Actually we will use
them explicitly only in the proof of Lemma \ref{smoothing.l}.

	\begin{remark} \label{semin vere}
One has that $f \in \sy m$ if and only if for all
	 integers $N_1$ and $N_2$  there exists a positive constant $C^m_{\delta,N_1, N_2}$ such that
	\begin{equation} \label{def seminorme}
\wp^m_{\delta,N_1, N_2}(f):=	\sup_{\begin{subarray}{c}
		\as  \in \R^d,\ k \in \Z^d,\\ \alpha \in \N^d,\ |\alpha| = N_1
		\end{subarray}}  \left|\partial_\as ^{\alpha}
        \hat{f}_k(\as )\right| |k|^{N_2} \langle \as \rangle^{-(m-\delta|\alpha|)}  <\infty\,,
	\end{equation}
with $\hat f$ the Fourier coefficients of $f$.
        \end{remark}
As anticipated in the notation of equation \eqref{def seminorme}, in
the following we will use the constants $\wp^m_{\delta,N_1,N_2}$ as
seminorms.

We come to the normal form procedure:
we look for a generating function $g$ that we want to use to transform
$H$ to a normal form plus a remainder decaying at infinity faster than
$\left|\as\right|^\tb$. If $g\in\sy\eta$, with a suitable $\eta$ (as
it will occur), the Lie transform $\Phi_g$ conjugates $H$ to 
\begin{equation}
  \label{main.eq}
h_0+P+\left\{h_0;g\right\}+{\rm lower\ order\ terms}\ .
\end{equation}
So we look for a symbol $g$ s.t. $P+\left\{h_0;g\right\}$ is
in normal form. Actually we will construct a symbol $g$ with the
property that $P+\left\{h_0;g\right\}$ consists of a part in normal
form plus a part decaying at infinity faster than any inverse power of
$|\as |$.

\subsection{Solution of the Cohomological equation} \label{cutoffs} 

In this subsection we are going to prove the lemma of solution of the
Cohomological  equation
	\begin{equation}\label{solve.me}
\left\{h_0,g\right\}+f-Z\in\sy{-\infty}\ .
	\end{equation}
        It is
a small variant of Lemma 5.8 of \cite{BLM19}, we give the proof for
the sake of completeness.  
 
\begin{lemma}\label{lem.hom}
Let $\frac{2}{3}<\delta<1$, then the following holds true: $\forall
f\in\sy m$, there exist $g\in\sy{m-\delta}$, $Z\in\sy{m}$, with $Z$ in
normal form, s.t. \eqref{solve.me} holds.
\end{lemma}

First, following \cite{BLM19} we split $f$ in a resonant, a nonresonant and a smoothing
part. This will be done with the help of suitable cutoffs, so
{let $\chi\in C^{\infty}(\R,\R)$ be a symmetric cutoff function which
  {is equal to 1 in $[-\frac{1}{2},\frac{1}{2}]$ and has support in
    $[-1,1]$.}
With its help we define,
	\begin{align}
	\label{cut-off-piccoli-divisori.1}
	\tilde\chi_k( \as) &:= \chi
	\left(\frac{\|k\|}{\|\as\|^{\mu}}\right)\ ,\qquad
	\chi_k( \as) := \chi\left({\frac{\as \cdot
				k}{\| \as\|^{ \delta}\|k\|}}  \right)\ ,
	\\
	\label{cut-off-piccoli-divisori.4}
	d_k( \as) &:= \frac{1}{\as \cdot
          k}\left(1-\chi\left({\frac{a \cdot
          		k}{\| \as\|^{ \delta}\|k\|}}\right)\right)\,.
	\end{align}
} By a simple computation one verifies that such functions are
symbols, precisely (for more details see Lemma 5.4 of \cite{BLM19}) $
\chi_k,\ \tilde\chi_k \in S_\delta^0 $, and $ d_k \in
S_\delta^{-\delta} $.  We use the above cutoffs
to decompose any function $f \in \sy m$:
\begin{equation}\label{splitting op}
f = f^{(nr)} + f^{(res)} + f^{(S)}\,.
\end{equation}
with
\begin{align}
\label{split.2}
& f^{(nr)}(\as,\fhi):= \sum_{k \in Z^d \setminus \{ 0 \}}(1-\chi_k(\as )) \tilde\chi_k(\as )
\hat{f}_k(\as)e^{ik\fhi}\ ,
\\
\label{split.1}
&	f^{(res)}(\as,\fhi):= \sum_{k \in Z^d }\chi_k(\as )
\tilde\chi_k(\as )\hat{f}_k(\as)e^{ik\fhi} \ ,
\\
\label{split.3}
& f^{(S)} (\as,\fhi):= \sum_{k \in Z^d \setminus \{ 0
	\}}(1-\tilde\chi_k(\as ))\hat{f}_k(a)e^{ik\fhi}\ .
\end{align}
Furthermore, one has that,
if $f \in \sy m$, then $f^{(res)}, f^{(nr)} \in \sy
        m$ and $f^{(res)}$ is in normal form.
Concerning $f^{(S)}$, by a variant of Lemma 5.6 of \cite{BLM19} one
has the following Lemma, whose proof we recall for the sake of
completeness.

\begin{lemma}
	\label{smoothing.l}
	Assume $f\in\sy m$, then $f^{(S)}\in \sy{-\infty}$.
\end{lemma}
\proof Consider the 
$k$-th Fourier coefficient of $f^{(S)}$: keeping into account that
$(1-\tilde\chi_k(\as ))$ is supported in the region $
\left\|k\right\|\geq\left\| \as \right\|^\mu$,
one has 
\begin{align*} 
\left|(1-\tilde\chi_k(\as ))\hat f_k(\as)\right|
\leq \frac{\left|(1-\tilde\chi_k(\as ))\hat f_k(\as)\right|}{\left\|
  k\right\|^N}\left\| k\right\|^N
\\
\sleq \frac{\left|\hat f_k(\as)\right|}{\left\|
  \as\right\|^{\mu N}}\left\| k\right\|^N \leq
\wp^m_{\delta,0,N}(f)\langle \as  \rangle^{m}\frac{1}{\langle \as 
  \rangle^{N\mu}}\ ,
\end{align*}
which, provided $N$ is large enough decreases at infinity as much as
desired. The control of the other seminorms is done similarly and is omitted.
\qed

\noindent
{\it Proof of Lemma \ref{lem.hom}}.  Define
\begin{equation}
\label{G0}
g:=i\sum_{k\not=0}{d_k(\as )}\hat f_k(a)\ ,
\end{equation}
then it is immediate to verify that
$$
\left\{h_0;g\right\}\equiv
-\as\cdot\frac{\partial}{\partial\fhi}g=f^{(nr)}\ .
$$
\qed

\subsection{End of the proof of Theorem \ref{norm.form}}\label{finedim1}

In this subsection we prove the following iterative lemma from which Theorem
\ref{norm.form} immediately follows.

\begin{lemma}\label{norm.form.lemma} Fix $M$,
	let $H$ be as in equation \eqref{hami}.
	There exists $0<\delta_* <\mom$ and $\mu_*>0$ such that, if
        $\delta_*<\delta<\mom$, and $0<\mu<\mu_*$, define
	\begin{equation}\label{sigma.1}
	\da:=\min\left\{2\delta-\tb;\delta
	\right\}\ ;
	\end{equation} 
        then $\da>0$ and the following holds.  For any 
        $\forall n\in\N$ with $M\geq n \geq 0$ there exists a time dependent
    canonical transformations $\cT_n$ conjugating $H$ to
	\begin{align}
	\label{eq.in.forma}
	H_n=h_0+Z_n+  R_n+\widetilde R_n\,,
	\end{align}
	where $Z_n \in \sy \tb$ is  in normal form;
 $R_n \in \sy {\tb- n\da}$, $\widetilde R_n \in \cR^N$. Furthermore, denoting as before $(\as,\fhi)=\cT_n(\as ',\fhi',t)$,
                  one has $\left\|\as-\as'\right\|\leq C_n
                  \left\|\as\right\|^{\tb-\delta }$.  
\end{lemma}
\proof
We {prove the theorem by induction.} In the case $n=0$, the claim is
trivially true.  

\def\rar{S_\delta}

We consider now the case $n>0$. 
Denote $m:=\tb-n\da$; we determine $g_{n+1}\in\rar^{\eta}$,
$\eta=\tb-n\da-\delta<1$, according to Lemma \ref{lem.hom} with $f$
replaced by $R_n$. Then one uses $\Phi_{g_{n+1}}$ to conjugate $H_n$ to
$H_n'$ given by
\begin{equation}\label{passaggetti}
\begin{aligned}
H_n'&=H-\left\{H_n;g_{n+1}\right\}+\rar^{m+2(\eta-\delta)}+\rar^{\eta}+\widetilde{\widetilde
{R}}_{n+1}+\tilde R_n\circ\Phi_{g_{n+1}}
\\
&=
h_0+Z_n+R_n-\left\{h_0;g_{n+1}\right\}+\rar^{\tb+\eta-\delta}+
\rar^{m+2(\eta-\delta)}+\rar^{\eta} +\widetilde{\widetilde
{R}}_{n+1}+\tilde R_n\circ\Phi_{g_{n+1}}
\\
&=   h_0+Z_n+R_n^{(res)}+\rar^{m-(2\delta+\delta-\td)}+\rar^{\tb+\eta-\delta}+
\rar^{m+2(\eta-\delta)}+\rar^{\eta}
\\
&\null\qquad+\widetilde{\widetilde
{R}}_{n+1}+\tilde R_n\circ\Phi_{g_{n+1}} \ ,
\end{aligned}
\end{equation}
where the term $\widetilde{\widetilde
{R}}_{n+1}$ contains the remainder of the expansion of the Lie transforms of
the different functions. 
Define now $Z_{n+1}:=Z_n+R_n^{(res)}$, $\tilde R_{n+1}:=\widetilde{\widetilde
{R}}_{n+1}+\tilde R_n\circ\Phi_{g_{n+1}}$ and $R_{n+1}$ to be the sum of the
remaining terms.  Writing explicitly the
different exponents of the classes $S$ of the terms composing $R_{n+1}$, we get that they are given by
\begin{align*}
e_1:=\tb-n\da-(3\delta-\td)=\tb-n\da-\da_1\ ,\quad
\da_1:= 3\delta-\td
\\
e_2:=\tb+\tb-n\da-2\delta=\tb-n\da-\da_2 \ ,\quad
\da_2:=2\delta-\tb\ ,
\\
e_3:=\tb-n\da+2(\tb-n\da-2\delta)=\tb-n\da-\da_3 \ ,\quad
\da_3:=2(n\da+2\delta-\tb)
\\
e_4:=\tb-n\da-\delta=\tb-n\da-\da_4\ ,\quad
\da_4:=\delta\ .
\end{align*}
Remarking that $\da_3\geq\da_2$ and taking the smallest
$\da$ one immediately gets the thesis. \qed

To conclude the proof of Theorem \ref{norm.form} just take $M=[N/\da]+1$.

\section{Geometric Part}\label{geometric}

\subsection{The partition}

Following Nekhoroshev, in this section we partition the action space
$\R^d$ according to the resonance relations fulfilled in each
region. We adapt the construction to our global setting. The
construction is very similar to the one developed in a quantum
context in \cite{BLM22a,QN}.
As in original Nekhoroshev's construction, the sub moduli of $\Z^d$
play a fundamental role in this construction, so, we first  recall
their definition.

\begin{definition}
  \label{modulo}
  A subgroup $M \subseteq \Z^d$ will be called a module if 
$\displaystyle{ \Z^d \cap \textrm{span}_\R M= M}$. Given a module
  $M$, we will denote $M_{\R}$ the linear subspace of $\R^d$ generated
  by $M$. Furthermore, given a vector $\as\in\R^d$ we will denote by
  $\as_M$ its orthogonal projection on $M_{\R}$.
  \end{definition}

In order to perform our construction we take positive parameters
$\delta,$  $ \mu,$  $\tC_1, \dots, \tC_d$\,, $\tD_1, \dots, \tD_d$, $\tR$ 
fulfilling \begin{equation}
\label{parameters}
\begin{gathered}
  \frac{d(d+1)}{2}\mu<1-\delta\,,
  \\
1 = \tC_1 < \tC_2 \cdots < \tC_d\,,\\
1 = \tD_1 < \tD_2 \cdots < \tD_d\, ,
\end{gathered}
\end{equation}
and define
\begin{equation}
  \label{idelta}
\delta_s:=\delta+\frac{s(s-1)}{2}\mu\left(=\delta_{s-1}+(s-1)\mu\right)\ ,
  \end{equation}
while $\tR$ will be assumed to be large enough. 

We start by the following definition

\begin{definition}[Resonant zones]\label{RZ}
	Let $M$ be a module of $\Z^d$ of dimension $s$.
	\begin{itemize}
		\item[(i)] If $s = 0$, namely $M = \{0\},$ we say that
                  $\as\in\cZ^{(0)}_{M}$ if either
                  $\left\|\as\right\|<\tR$ or
                  \begin{equation}
                    \label{riso}
\left|\as\cdot k\right|\geq \left\| k\right\|\left\|\as\right\|^\delta\ ,\qquad \forall
k\ :\ \quad \left\|k\right\|\leq \left\|\as\right\|^\mu\ .
                    \end{equation}	
$\cZ^{(0)}_{ \{0\}}$ will be called the \emph{non resonant zone}.
                \item[(ii)] If $s\geq 1$, for any set of linearly
		independent vectors $\{k_1, \dots, k_s\}$ in $M$, we
say that $\as\in\cZ_{k_1,...,k_s}$ if $\left\|\as\right\|\geq \tR$ and
$\forall j=1,...,s$ one has
\begin{equation}
  \label{small.j}
\left\|\as\cdot k_j\right\|\leq\tC_j\left\|
k\right\|\left\|\as\right\|^{\delta_j} \ \quad {\rm and}\quad
\left\|k\right\|\leq\tD_j \left\|\as\right\|^\mu\ . 
  \end{equation}
then we put
		\begin{equation}
		\cZ^{(s)}_{M} := \bigcup_{\begin{subarray}{c}
			k_1, \dots, k_s \\ \textrm{lin. ind. in } M
			\end{subarray} } \cZ_{k_1, \dots, k_s} \,.
		\end{equation}
	\end{itemize}
	The sets $\cZ^{(s)}_M$ are called \emph{resonant zones}.
\end{definition}

The sets $\cZ^{(s)}_{M}$ contain points $\as$ which are in
resonance with \textit{at least} $s$ linearly independent vectors in
$M$.

\begin{remark} \label{rmk.inscatolate}
	Fix $r, s \in \{1, \dots, d\}$ with $1\leq r<s$, then for any $M$ with dim $M=s$, one has
	$$
	\cZ^{(s)}_{M} \subseteq \bigcup_{ \begin{subarray}{c} M^\prime \subset M\\
		\textrm{dim} M^\prime = r 
		\end{subarray}} \cZ^{(r)}_{M^\prime}\,.
	$$
\end{remark}

Following Nekhoroshev we now define the \emph{resonant blocks}, which are composed by the points which are resonant with
the vectors in a module $M,$ but are non-resonant with the vectors
$k\not \in M$ and the \emph{extended blocks} which will turn out to be
invariant under the dynamics of $h_0+Z$.

\begin{definition}[Resonant blocks]\label{def blocchi} We first
  define, for $M=\Z^d$, the set $\cB^{(0)}_{\Z^d}:=\cZ^{(0)}_{\Z^d}$. Then, we proceed
  iteratively: for $s<d$ let $M $ be a module of dimension $s$, we define the
  resonant block
	$$
	\cB^{(s)}_M = \cZ^{(s)}_M \backslash \left( \bigcup_{\begin{subarray}{c} s^\prime > s \\ \dim M^\prime = s^\prime \end{subarray}} \cB^{(s^\prime)}_{M^\prime}\right)\,.
	$$
\end{definition}

\begin{definition}[Extended blocks and fast drift planes] \label{def blocchi estesi}
For any module $M$ of dimension $s,$ we define
$$ \widetilde{E}^{(s)}_M = \left\{ \cB^{(s)}_M + M_{\R}\right\} \cap
\cZ^{(s)}_M
$$ and the \emph{extended blocks}
$$
	E^{(s)}_{M} =  \widetilde{E}^{(s)}_M \backslash \left( \bigcup_{\begin{subarray}{c}
		s^\prime < s \\ \dim M^\prime = s^\prime
		\end{subarray}} E^{(s^\prime)}_{M^\prime}\right)\,,
	$$
	where $\displaystyle{A + B = \left\lbrace a + b\ |\  a \in A, b \in B \right \rbrace}\,.$
	Moreover, for all $p \in E^{(s)}_{M}$ we define the \emph{fast drift plain}
	$$
	\Pi^{(s)}_{M}(p) = \left\lbrace p + M_{\R} \right \rbrace \cap
        \cZ^{(s)}_{M}\,\ .
	$$
\end{definition}
\subsection{Properties of the partition}

A useful technical tool is given by the following remark:
\begin{remark}\label{rmk.a.or.b}
	If $\as, b \in \R^d$, fulfil ${\|\as\|, \|b\| \geq 1}$, and
	$$
	\| \as - b\| \leq C {\| b\|}^{\tilde \delta}\,,
	$$
        with some constants
        $C>0$ and $0< \tilde\delta < 1$, 
	then one has
	$$
	\| \as - b\| \lesssim {\| \as\|}^{\tilde\delta}\,.
	$$
\end{remark}

We start now to study the properties of the partition.

\begin{remark}
  \label{non res}
By the very definition of $\cZ^{(s)}_M$, for any $s\geq 1$, one has $\cZ^{(s)}_M\cap B_{\tR}(0)=\emptyset$.  
\end{remark}

\begin{lemma}\label{lemma.bd.1}
	Provided $\tR$ is large enough, the resonant zone
	$\cZ^{(d)}_{\Z^d}$ is empty.
\end{lemma}

The proof requires the use of the following Lemma from \cite{GioPisa}. For
the proof we refer to \cite{GioPisa}.

\begin{lemma}\label{lemma.giorgilli} [Lemma 5.7 of \cite{GioPisa}]
	Let $s \in \{1\,, \dots\,, d\}$ and let $\{ u_1\,, \dots u_s
	\}$ be linearly independent vectors in $\R^d\,.$ Let $w \in
	\Span{\{ u_1\,, \dots u_s\}}$ be any vector. If $\alpha\,, N$
	are such that
	\[
	\begin{gathered}
	\norm{u_j} \leq N \quad \forall j = 1\,, \dots s\,,\\
	| \scala{w}{u_j} | \leq \alpha \quad \forall j = 1\,, \dots s\,,
	\end{gathered}
	\]
	then
	\[
	\norm{w} \leq \frac{s N^{s-1} \alpha }{\textrm{Vol} \{u_1\,| \cdots\,| u_s \}}\,.
	\]
\end{lemma}

\begin{proof}[Proof of Lemma \ref{lemma.bd.1}]
	Assume that $\cZ^{(d)}_{\Z^d}$ is not empty and take $\as \in
	\cZ^{(d)}_{\Z^d}$, then there exist $\{{k}_1, \dots,
	   {k}_{d}\} \subset \Z^d$ linear independent vectors such
           that \eqref{small.j} is fulfilled by $\as$ with the given
           $k_j$'s. 
Using Lemma \ref{lemma.giorgilli} we deduce
	$$
	\| \as\| \leq d (\tD_d)^{d} \tC_d {\|\as\|}^{\delta_d  + \mu d}\,.
	$$
By \eqref{parameters}, one has that $\delta_d+d\mu <
1$. So, provided $\tR$ is large enough this is in contradiction  $\left\|\as\right\|<\tR$.
\end{proof}

\begin{lemma}\label{lemma.diametri}	
There exists a constant $C$ s.t. if $\Pi^{(s)}_M(\as)$ is a fast drift
plane, then 
\begin{equation}
  \label{diam}
diam(\Pi^{(s)}_M(\as)) \leq C {\|\as\|}^{\delta_{s+1}}\,.
  \end{equation}
\end{lemma}
\proof First, by definition of resonant zones, for
$a\in\cZ^{(s)}_M$, there exist $k_1,...,k_s\in M$ s.t.  $|\as\cdot k_j|\leq
\tC_s \left\| k_j \right\| {\|\as\|}^{\delta_s}$, $\forall
j=1,...,s$, so that, by Lemma \ref{lemma.giorgilli}
	\begin{equation}
	\label{proi.1}
	\|\Pi_M \as\| \sleq  {\|\as\|}^{\delta_s+s\mu}\,.
	\end{equation}
       If $\as'\in \Pi^{(s)}_M(\as)$ then the same holds for
       $\as'$. So we have
        $$
\left\|\as-\as'\right\|=\left\|\Pi_M(\as-\as')\right\|\leq
\left\|\Pi_M\as\right\|+\left\|\Pi_M\as'\right\|\sleq
(\left\|\as\right\|^{\delta_{s+1}}+\left\|\as'\right\|^{\delta_{s+1}})\ .
        $$
By Remark \ref{rmk.a.or.b} this implies the thesis. \qed

In particular we have the following Corollary
\begin{corollary}
  \label{distB}
  If $\as\in E^{(s)}_M$ there exists $\as'\in \cB^{(s)}_M$ s.t.
  \begin{equation}
    \label{distbe}
\left\|\as-\as'\right\|\leq C\left\|\as\right\|^{\delta_{s+1}}\ .
  \end{equation}
  \end{corollary}
Indeed, by definition of extended block $\exists
\as'\in\Pi^{(s)}_M(\as)\cap \cB^{(s)}_M$, and therefore, by \eqref{diam} the corollary
holds.  

The next lemma ensures that, if the parameters $\tC_{j}, \tD_j$ are
suitably chosen, an extended block $E^{(s)}_{M, j}$ is separated from
every resonant zone associated to a module $M^\prime$ with dim$(M)=s'\leq s$,
which is not contained in $M.$ This is the extension to our context of
the classical property of separation of resonances.

\begin{lemma} \label{separo}[Separation of
	resonances] Take $K>0$.  There exist positive constants $
  \tR$, $\tilde{\tC}_{s+1}$ and $\tilde{\tD}_{s+1}$ depending only on
  $\mu, \delta_{s}, \tC_{s}, \tD_{s}, K$ such that, if
	$$
	\tC_{s+1} > \tilde{\tC}_{s+1}\,, \quad \tD_{s+1} >\tilde{\tD}_{s+1}, \quad \tR > \bar \tR\,,
	$$
	then the following holds true. Let $\as \in E^{(s)}_{M}$ for
	some $M$ of dimension $s= 1, \dots, d-1$, and let $\as^\prime \in \R^d\null$
	be such that 
	\begin{gather*}
\left\|\as-\as'\right\| \leq K{\| \as\|}^{\delta_{s+1}}\,,
	\end{gather*}
	then $\forall M^\prime \not\subset M$ s. t.  $s':=\dim M^\prime {\leq} s$ one
	has 
$$\as^\prime \notin \cZ^{(s')}_{M^\prime}\,.
$$
\end{lemma}
\begin{proof}
Assume by contradiction that $\as'\in \cZ^{(s')}_{M^\prime}$ for some $M^\prime \neq M.$ It follows
	that there exist  $s'$ integer vectors, $k_1,...,k_{s'}\in M'$ {\it
		among which at least one does not belong to $M$,} s.t.
	\begin{equation}
	\label{bei.11}
	\left|\scala{\as'}{k_j}\right|\leq
	\tC_{j}{\| \as' \|}^{\delta_j}{\norm{k_j}}\ ,\quad \norm{k_j}\leq
	\tD_{j}{\| \as'\|}^{\mu}\ . 
	\end{equation}
	Let $k_{\bar\j }$ be the vector which does not belong to
        $M$. By Corollary \ref{distB}, there exists $b\in \cB^{(s)}_{M}$ s.t. $\norm{\as -b}\sleq {\|\as\|}^{\delta_{s+1}}$ and
        thus also $\norm{\as' -b}\sleq {\|\as'\|}^{\delta_{s+1}}$ (of
        course with a different constant). Thus it follows that there
        exist  constants $\tilde \tC_{s+1}$, $\tilde \tD_{s+1}$ s.t. 
$$
	\left|b\cdot k_{\bar\j }\right|\leq
	\tilde{\tC}_{s+1}{\| b\|}^{\delta_{s+1}}\norm{k_{\bar{\j}}} \ ,\quad
	\norm{k_{\bar\j}}\leq \tilde{\tD}_{s+1}{\| b\|}
	^\mu\ .
	$$
        But, if $\tC_{s+1}> \tilde{\tC}_{s+1}$, $\tD_{s+1}> \tilde{\tD}_{s+1}$, this means that $b$ is also
	resonant with $k_{\bar\j }\not\in M $, and this contradicts the fact
	that $b\in \cB^{(s)}_{M, j}$.
\end{proof}

In order to take into account the effects of the remainder in the
normal form theorem and to
conclude the proof we need to extend 
the resonant planes.
\begin{definition}
  \label{extenp}
We define
\begin{align}
    \label{boh}
    \left[\Pi^{(s)}_M(\as)\right]^{ext}:=\bigcup_{\as'\in
      \Pi^{(s)}_M(\as)}B_{\left\|\as'\right\|^{\delta_{s+1}}}(\as')\ ,
\\
    \label{boh1}
    \left[\Pi^{(s)}_M(\as)\right]^{ext}_{tr}:=    \left[\Pi^{(s)}_M(\as)\right]^{ext}\cap\cZ^{(s)}_M\ ,
\end{align}
\begin{equation}
    \end{equation}
where, ad before $B_R(\as)$ is the ball of radius $R$ and centre
$\as$.
\end{definition}

\begin{remark}
  \label{diam.rm}
  One has
  \begin{equation}
    \label{diam.ext}
diam(\left[\Pi^{(s)}_M(\as)\right]^{ext} )\leq C_s\left\|\as\right\|^{\delta_{s+1}}
\ ,    \end{equation}
  for some $C_s$.
\end{remark}

\begin{remark}
  \label{esteso.1} 
If the constants $
  \tR$, $\tilde{\tC}_{s+1}$ and $\tilde{\tD}_{s+1}$ are chosen
  suitably, then $\forall \as'\in\exte$ and all $\as''$ s.t.
  $$
\left\|\as''-\as'\right\|\leq \left\|\as'\right\|^{\delta_{s+1}}\ ,
$$
one has \begin{equation}
  \label{nonrisuono}
\as''\not\in\cZ^{s'}_{M'}\ ,\quad \forall (s'\leq
s\ ,\ M':M'\not\subseteq M)\ .
  \end{equation}
\end{remark}
\begin{remark}
  \label{esco}
By Lemma \ref{separo}, it follows that if
$\as''\in\partial\Pi^{(s)}_M(\as')$, $\as'\in\exte_{tr}$, then
$\as''\in\partial 
E^{(s')}_{M'}$ with $M\subset M'$ and $s'<s$. 
\end{remark}
\begin{lemma}
  \label{non.risuono}
If the constants $
  \tR$, $\tilde{\tC}_{s+1}$ and $\tilde{\tD}_{s+1}$ are chosen
  suitably, then, $\forall \as\in[\Pi^{(s)}_M(p')]^{ext}$ and $\forall
  k\not\in M$ one has
  \begin{equation}
   \left\|k\right\|\leq
   \left\|\as\right\|^\mu\ \Longrightarrow\ \left|\as\cdot
   k\right|\geq \left\|\as\right\|^\delta\left\|k\right\|\ .
    \end{equation}
\end{lemma}
\proof Following the proof of Lemma \ref{separo}, assume by
contradiction that $\left|\as\cdot k\right|\leq
\left\|\as\right\|^{\delta}\left\|k\right\|$, then $\exists \as'\in
E^{(s)}_M$ s.t. $\left|\as-\as'\right|\sleq
\left\|\as\right\|^{\delta_{s+1}}$, and therefore $\exists \as''\in
\cB^{(s)}_M$ s.t. $\left|\as-\as''\right|\sleq
\left\|\as\right\|^{\delta_{s+1}}$. It follows
$$
\left|\as''\cdot k\right|\leq
\left|\as-\as''\right|\left\|k\right\|+\left|\as\cdot k\right|\sleq
\left|\as\right| ^{\delta_{s+1}}\left\|k\right\|\sleq
\left|\as''\right| ^{\delta_{s+1}}\left\|k\right\|\ ,
$$
but if $\tC_{s+1}$ is chosen large enough, this means that $\as''$
fulfils \eqref{small.j} with $j=s+1$, against
$\as''\in\cB^{(s)}_M$. \qed 
\begin{corollary}
  \label{noncisono}
Consider the normal form $Z_N$ obtained by Theorem
\ref{norm.form}. Define
$Z(\as,\fhi,t):=Z_N(\as,\fhi,t)(1-\chi\left(\frac{\left\|\as\right\|^2}{\tR}\right))$
(which is supported outside $B_{\tR}(0)$).
For $\as'\in\exte$,
\begin{equation}
  \label{non.co}
  Z(\as',\fhi,t)=\sum_{k\in M}Z_{k}(\as',t)e^{\im k\cdot \fhi}\ ,
\end{equation}
namely the sum is restricted to $k\in M$. 
\end{corollary}
\begin{remark}
  \label{unico}
$\forall \as\in\R^d$ $\exists ! M$ s.t. $p\in E^{(s)}_M$. The
  important point is the unicity. 
\end{remark}
We are now ready to prove the following Theorem, giving a control on
the dynamics over long times. This is the typical Nekhoroshev type
theorem adapted to our $C^{\infty}$ context.
\begin{theorem}
  \label{nekho}
There exist positive $K_1<K_2<...<K_d$ s.t. the following
  holds true: consider the Cauchy problem for the Hamiltonian system
  \eqref{forma.normale} with initial datum $\as_0$. Let $M$ with $dim
  M=s$ be s.t. $\as_0\in E^{(s)}_M$. Then one has
  \begin{equation}
    \label{stoqui}
\as(t)\in\exten{\as_0}\ ,\quad \forall |t|\leq
\frac{1}{K_s}\left\|\as_0\right\|^{N+\delta}\ ,
    \end{equation}
  and thus, in particular, for the same times one has
  \begin{equation}
    \label{stopqui}
\left\|\as(t)\right\|\leq 2\left\|\as_0\right\|\ .
    \end{equation}
\end{theorem}
\proof First we remark that this is true when $s=0$ and thus $\as_0$
is in the nonresonant region. Indeed, in this case one has that the
equations for $\as$ reduce to $\dot\as=-\frac{\partial
  R^{(N)}}{\partial\fhi}=\cO(\left\|\as_0^{-N}\right\|)$.

Following Nekhoroshev we proceed by induction on $s$. So, assume
$\as_0\in E^{(s)}_M$ with $s\geq 1$. Assume that the result has been
proved for $s-1$ and we prove it for $s$. Introduce in
$\exten{\as_0}$ coordinates $\as=(\as_M,\as_{\perp})$ with $\as_M\in
M_{\R}$ and $\as_\perp\in
(M_{\R})^\perp$. In these coordinates one has that $\Pi^{(s)}_M(\as)$ is
the set of $\tilde\as=(\tilde\as_M,\tilde\as_{\perp})$ with $\tilde
\as_\perp=\as_\perp$ and $\tilde \as_M$ belonging to some domain $D$ (which
depends also on $\as_\perp$, but this is not important in the
following). Then one has that $\tilde\as\in \partial \Pi^{s}_M(\as)$
is equivalent to $\tilde \as_M\in\partial D$. It follows that $\tilde
\as\in\partial [\Pi^{(s)}_M(\as_0)]^{ext}_{tr}$ implies that either
$\left\|\as_\perp\right\|=\left\|\as_0\right\|^{\delta_{s+1}}$ or
$\tilde \as\in\partial\Pi^{(s)}_M(\as')$ for some
$\as'\in\exten{\as_0}_{tr}$.

We are now ready to conclude the proof: assume that $\exists
0<\bar t\leq \left\|\as_0/K_s\right\|$ s.t. $\as(\bar t)\in\partial
\exten{\as_0}_{tr}$. If it does not exists, then there is nothing to
prove. Then either $\left\|\as_{0,\perp}-\as_{\perp}(\bar
t)\right\|=\left\|\as_0\right\|^{\delta_{s+1}}$ or $\as(\bar
t)\in\partial \Pi^{(s)}_M(\as')$ for some
$\as'\in\exten{\as_0}_{tr}$. Now the first possibility is ruled out by
the fact that in $\exten{\as_0}$ the equations for $\as_\perp$ reduce
to $\dot\as_{\perp}=\cO(\left\|\as_0\right\|^N)$. So, assume that
$\as(\bar t)\in \partial \Pi^{(s)}_M(\as')$, then, by Remark
\ref{esco} one has $\as(\bar t)\in \partial E^{(s')}_{M'}$ with $s'<s$. Then,
by induction, considering $\as(\bar t^+)$ as an initial datum, one gets
that (by Remark \ref{diam.rm})
\begin{equation}
  \label{perdavvero}
\left\|\as_0-\as(\bar t)\right\|\leq C_s \left\|\as(\bar
t)\right\|^{\delta_s}\ ,\quad \left\|\as(t)-\as(\bar t)\right\|\leq
C_s\left\|\as(\bar t)\right\|^{\delta_s}\ , \quad \left|t-\bar
t\right|\leq \frac{\left\|\as(\bar t)\right\|^{N+\delta}}{K_{s-1}}\ .
  \end{equation}
But, actually, by the bound on time \eqref{stoqui}, the times fulfil
$$
\left|t-\bar t\right|\leq |t|+|\bar t|\leq
\frac{2\left\|\as_0\right\|^{N+\delta}}{K_s}\leq C
\frac{2\left\|\as(\bar t)\right\|^{N+\delta}}{K_s}\ ,
$$
and therefore, if $2C/K_s<K_{s-1}$, the estimates \eqref{perdavvero}
hold for the times we are interested in. Concerning the distance from
$\Pi_{M}^{s}(\as_0)$, the above estimates imply 
$$
\left\|\as(t)-\as_0\right\|\leq 2\tilde C_s \left\|\as(\bar
t)\right\|^{\delta_s}\sleq \left\|\as_0\right\|^{\delta_s}\ , 
$$
so the left hand side is smaller than
$\left\|\as_0\right\|^{\delta_{s+1}} $ provided $\as_0$ is large
enough. For $\as_0$ in a compact set the estimate is trivial. \qed

\noindent{{\it Proof of Theorem \ref{main}}}. Assume that there exists
a solution with
$$
\limsup_{t\to+\infty}\left\|\as(t)\right\|=+\infty\ ,
$$
otherwise the result holds trivially. Let $R_k:=R_02^{k}$ with
$R_0:=\left\|\as_0\right\|$, then there exists a sequence of times
$t_k$ s.t.
$$
\sup_{|t|\leq t_k}\left\|\as(t_k)\right\|=2R_k\ . 
$$
Applying Theorem \eqref{nekho} with initial datum $\as(t_k)$, one gets
$$
t_{k+1}>\frac{1}{K_d}R_k^{N}+t_k
$$
(where we redefined $N+\delta\to N$). So, taking $t_0=0$, one gets
\begin{equation}
t_{L+1}\geq \sum_{k=0}^L\frac{1}{K_d}(R_02^k)= \frac{1}{K_d}R_0^N
\frac{2^{(L+1)N}-1}{2^N-1}\geq \frac{1}{2K_d}R_0^N2^{LN}\ .
  \end{equation}
Thus, defining
$$
\tau_0:=0\ ,\quad \tau_{k+1}:=\frac{1}{2K_d}(R_02^k)^N\ ,\quad k\geq
0 \ ,
$$
we have
\begin{equation}
  \label{sup1}
\sup_{|t|\leq\tau_k}\left\|\as(t)\right\|\leq 2R^k\ .
  \end{equation}
To write a global in time formula consider the function
$$
\widetilde \theta(t):=\prod_{k=0}^\infty\left(2\theta(t-\tau_k)\right)\ ,
$$ with $\theta(t)$ the standard Heaviside step function.  Remark that
$\widetilde \theta$ is well defined since, for any time only a finite
number of factors is different from zero. Thus we have the global
estimate
$$
\left\|\as(t)\right\|\leq R_0\widetilde \theta(t)\ .
$$
Consider now the function $f(t):=\frac{R_0\widetilde
  \theta(t)}{t}$. We consider it only for $t\geq1$. Such a
function has positive jumps at $\tau_k$ and in all the other intervals
it is monotonically decreasing like $t^{-1}$. In order to bound such a
function we look for a function interpolating the peaks. One has, for
$k\geq 1$
$$
f(\tau_k^+)=2K_d\frac{R_02^{k+1}}{R_{k-1}^N}=2K_d\frac{4
  R_{k-1}}{R_{k-1}^N}
=\frac{2\left({2K_d}\tau_k\right)^{1/N}} {\tau_k}\ .
$$
So it is clear that an interpolating function is
\begin{equation}
  \label{tildef}
\tilde f(t):=4\left({2K_d}\right)^{1/N}\frac{t^{1/N}}{t}\ ,
\end{equation}
and in this way, $\forall t\geq\tau_1$, one has $f(t)\leq \tilde
f(t)$. It follows
\begin{equation}
  \label{prima.stima}
\left\|\as(t)\right\|\leq
4\left({2K_d}\right)^{1/N}t^{1/N}\ ,\quad t\geq\tau_1\ .
\end{equation}
We now manipulate such an expression to get the thesis. Taking into
account that at $\tau_1$ the r.h.s. of \eqref{prima.stima} is equal to
$4R_0$, one has, $\forall t>0$
\begin{align}
  \label{ultima.stima}
\left\|\as(t)\right\|\leq \max\left\{4R_0,4
\left({2K_d}\right)^{1/N} t^{1/N}\right\}
  =\max \left\{ 4R_0,4
 \left( {2K_d}\right) ^{1/N} \tau_1^{1/N} 
\left( \frac{t}{\tau_1} \right)^{1/N} \right\}
  \\
  =4R_0 \max\left\{1,\left(\frac{t}{\tau_1}\right)^{1/N} \right\}\leq
    4R_0\left\langle \frac{t}{\tau_1}\right\rangle^{1/N}\ .
\end{align}

Still we have to take into account the change of coordinates. In this
way we get that there exists $\bar R_0$, s.t., if the initial datum
fulfills $\left\|\as_0\right\|\geq \bar R_0$, then one has 
$$
\left\|\as(t)\right\|\leq  16R_0\left\langle
\frac{t}{\tau_1}\right\rangle^{1/N}\ .
$$
 \qed

\bibliographystyle{alpha}
\bibliography{biblio}

\end{document}